\documentclass[journal]{IEEEtran}

\ifCLASSINFOpdf
\else

\usepackage[dvips]{graphicx}
\fi
\usepackage{url}
\usepackage{graphicx}
\usepackage{diagbox}
\usepackage{balance}
\usepackage{multirow, cite}
\usepackage{xcolor,colortbl, soul}
\usepackage{tabularx,booktabs}
\usepackage[utf8]{inputenc} 
\usepackage[T1]{fontenc}
\usepackage{url}
\usepackage{ifthen}
\usepackage{cite}
\usepackage[cmex10]{amsmath} 
\usepackage{subfigure}
\usepackage{cite}
\usepackage{amsmath,amssymb,graphicx,epsfig,fancyhdr,amsthm,tabulary,epstopdf}
\usepackage{algorithm}
\usepackage{algorithmic}
\usepackage{pdfsync}
\usepackage[justification=centering]{caption}
\newtheorem{thm}{Theorem}[]

\newtheorem{lem}{Lemma}

\theoremstyle{remark}

\theoremstyle{definition}

\newcommand{\CASE}[1]{\STATE \textbf{case} #1\textbf{:} \begin{ALC@g}}
	\newcommand{\ENDCASE}{\end{ALC@g}}

\newcommand{\DEFAULT}{\STATE \textbf{default:} \begin{ALC@g}}
	\newcommand{\ENDDEFAULT}{\end{ALC@g}}
\newcommand{\DEFAULTLINE}[1]{\STATE \textbf{default:} }

\begin{document}
\bstctlcite{IEEEexample:BSTcontrol}
\title{ 
	Secure Transmission in MIMO-NOMA Networks}

\author{Yue Qi, \IEEEmembership{Student Member, IEEE}, and Mojtaba Vaezi, \IEEEmembership{Senior Member, IEEE} 
}

\author{Yue~Qi,~\IEEEmembership{Student~Member,~IEEE,} 
	and~Mojtaba~Vaezi,~\IEEEmembership{Senior~Member,~IEEE} 
	\thanks{The authors are with the Department
		of Electrical and Computer Engineering, Villanova University, Villanova,
		PA 19085 USA (e-mail: yqi@villanova.edu; mvaezi@villanova.edu).}
}

\maketitle

\begin{abstract}
This letter focuses on the physical layer security over two-user multiple-input multiple-output (MIMO) non-orthogonal multiple access (NOMA) networks. A linear precoding technique is designed to ensure the confidentiality of the message of each user from its counterpart.
This technique first splits the base station power between the two users and, based on that, decomposes the secure MIMO-NOMA channel into two MIMO wiretap channels,  and designs the transmit covariance matrix for each channel separately. The proposed method substantially enlarges the secrecy rate compared to existing linear precoding methods and strikes a balance between performance and computation cost. Simulation results verify the effectiveness of the proposed method.

\end{abstract}

\begin{IEEEkeywords}
MIMO-NOMA, physical layer security, wiretap, precoding, GSVD.
\end{IEEEkeywords}

\IEEEpeerreviewmaketitle

\section{Introduction}

	In view of its potential to increase connectivity,  reduce latency, 
	and improve spectral efficiency, non-orthogonal multiple access (NOMA) has attracted tremendous 
	attention for fifth generation  and beyond wireless networks \cite{vaezi2019interplay}. 	In NOMA, the base station broadcasts the same signal to serve multiple users  over the same resources in time/frequency/code/space.
	Due to the broadcast nature of transmission, NOMA users are susceptible  to internal and external eavesdroppers.  Therefore, new aspects of \textit{physical layer security} need to be analyzed in NOMA networks.

To fulfill the security requirements of  single-antenna NOMA networks, existing security techniques such as cooperated relaying and jamming have been proposed \cite{liu2017enhancing, chen2018physical,zheng2018secure,arafa2019secure}.
In multiple-input, multiple-output (MIMO) NOMA networks,  other methods such as \textit{artificial noise} (AN)-aided transmission  and beamforming \cite{lv2018secure,feng2019beamforming,zeng2019securing} have been proposed to make communications less vulnerable to `external'  eavesdroppers. These solutions  are mostly to secure data transmission from external eavesdroppers.   	 
However, since a superimposed signal is transmitted to a group of legitimate users, 
an important question is whether NOMA users can communicate their messages confidentially, or legitimate `internal' users  
may compromise their security? 
Early works have proved that, in a two-user MIMO-NOMA network, 
	both users can transmit their messages  concurrently and confidentially via \textit{secret dirty-paper coding} (S-DPC) 
 \cite{liu2010multiple}.\footnote{There are other important information-theoretic models related to MIMO-NOMA security with external eavesdroppers \cite{vaezi2019noma}.} The 
	complexity of S-DPC is, however, not acceptable in practice. This motivates the
	development of low-complexity, fast solutions, such as designing linear 
	precoding. In \cite{fakoorian2013optimality}, 
	a linear precoder based on \textit{generalized singular value 
	decomposition} (GSVD) is designed using orthogonal parallel channel 
	transmission. However, this solution is far from the capacity region.
	In \cite{park2015weighted}, the secrecy capacity of two-user MIMO-NOMA 
	is transformed into a weighted 
	secrecy sum-rate maximization problem. This method improves the achievable secrecy rates but still is time-consuming.

	In this letter, we  design \textit{precoding} and \textit{power allocation} matrices to achieve the secrecy 
	capacity 
	region of two-user Gaussian MIMO-NOMA networks under the total  power 
	constraints with reasonable time consumption. 
	Due to the complexity of the problem, we decompose it into two wiretap channel sub-problems. 
Then, we design the transmit covariance matrix for each problem and show that the performance loss is negligible.

The contributions of this letter are summarized as follows:
\begin{enumerate}
	\item[$\bullet$] By splitting the  total power between the two users, we 	decompose the two-user secure MIMO-NOMA channel into two MIMO wiretap channels  with respective  power constraints.
	\item[$\bullet$] We propose new precoding and power allocation  to solve the new problems, and introduce  an efficient and  cost-effective  algorithm for secure transmission in MIMO-NOMA networks within practical ranges of antennas. 
	Our approach strikes a  balance between performance and 
	time consumption.
\end{enumerate}

\textit{Notations:} $\rm{tr}(\cdot)$ and $(\cdot)^{T}$ denote the trace and transpose 
of matrices. $\mathbb{E}(\cdot)$ denotes expectation. ${\rm diag}(\lambda_1, 
\dots, \lambda_n)$ represents diagonal matrix with diagonal elements  
$\lambda_1, 
\dots, \lambda_n$.
$\mathbf{Q} \succcurlyeq \mathbf{0} $ means $\mathbf{Q}$ is a positive 
semidefinite matrix, and $\mathbf{I}$ is an identity matrix.

\begin{figure}[t]
	\centering
	\includegraphics[width=0.35\textwidth]{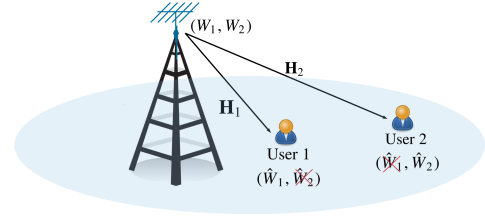}
	\caption{Illustration of a two-user \textit{secure} MIMO-NOMA network in which each user should be able to decode only its message.}
	\label{fig:1}
\end{figure}

\section{System Model}

We consider  a two-user MIMO-NOMA network, as 
shown in 
Fig.~\ref{fig:1}. The transmitter (Tx), user~$1$, and user~$2$ are equipped with $n_t$, $n_1$, and  $n_2$ antennas, respectively.  The  Tx
 serves the users with two confidential messages $W_1$ and $W_2$ (such as accessing bank accounts and performing online transactions), i.e., user~$i$ should not be able to decode $W_i$ when $i\neq j$, $i, j\in \{1, 2\}$. 
In this setting, user~$1$ can be seen as an eavesdropper to 
user~$2$ and vice 
versa. Due to NOMA transmission, the input vectors $x_1 \in \mathbb{R}^{n_t 
\times 1}$ and $x_2 \in \mathbb{R}^{n_t \times 1}$ intended for  user~$1$ 
and user~$2$ 
share the same time  and frequency slot.
\noindent The received signals at user~$1$ and user~$2$, respectively, are 
given by
\begin{subequations}\label{eq:signal model}
	\begin{align} 
	\mathbf{y}_1  &= \mathbf{H}_1(\mathbf{x}_1 +  \mathbf{x}_2) +\mathbf{w}_1,\\
	\mathbf{y}_2 &= \mathbf{H}_2(\mathbf{x}_1  + \mathbf{x}_2) + \mathbf{w}_2,
	\end{align}
\end{subequations}
in which $\mathbf{H}_1 \in \mathbb{R}^{n_1 \times n_t}$ and $\mathbf{H}_2 
\in \mathbb{R}^{n_2 \times n_t}$ are the channel matrices for 
user~$1$ and 
user~$2$, and $\mathbf{w}_1 \in \mathbb{R}^{n_1 \times 1}$ and 
$\mathbf{w}_2 \in \mathbb{R}^{n_2 \times 1}$ are independent identically 
distributed (i.i.d) Gaussian noise vectors  whose elements are zero mean and 
unit variance.

This setting is also known as MIMO broadcast channel (BC) with two 
confidential messages, and its  secrecy capacity region under the average total power constraint 
can be expressed as \cite{liu2010multiple, ekrem2012capacity} 
\begin{subequations} \label{eq: mathmodel}
	\begin{align}
	R_1 &\leq  \frac{1}{2}\log|\mathbf{I} + { \mathbf{H}_1 \mathbf{Q}_1 \mathbf{H}_1^T}|-\frac{1}{2}\log|\mathbf{I} +  \mathbf{H}_2 \mathbf{Q}_1 \mathbf{H}_2^T|, \label{eq: mathmodel_confi1} \\
	R_2 &\leq \frac{1}{2} \log \bigg|\mathbf{I} + \frac{\mathbf{H}_2 \mathbf{Q}_2 \mathbf{H}_2^T}{\mathbf{I} +\mathbf{H}_2 \mathbf{Q}_1 \mathbf{H}_2^T}\bigg|
	- \frac{1}{2}\log \bigg|\mathbf{I} + \frac{\mathbf{H}_1 \mathbf{Q}_2 \mathbf{H}_1^T}{\mathbf{I} +\mathbf{H}_1 \mathbf{Q}_1 \mathbf{H}_1^T}\bigg|  \label{eq: mathmodel_confi2} \\
	&\textmd{s.t.} \quad {\rm tr}(\mathbf{Q}_1+\mathbf{Q}_2) \leq P,\; \mathbf{Q}_1\succcurlyeq \mathbf{0},\;\mathbf{Q}_2\succcurlyeq \mathbf{0}\label{eq: mathmodel_const}
	\end{align}
\end{subequations}
in which $\mathbf{Q}_1 = \mathbb{E}(\mathbf{x}_1\mathbf{x}_1^{T})$ and $\mathbf{Q}_2 = \mathbb{E}(\mathbf{x}_2\mathbf{x}_2^{T})$ are the input 
covariance matrices corresponding to $\mathbf{x}_1$ and 
$\mathbf{x}_2$, respectively.

The capacity region in  \eqref{eq: mathmodel}
is obtained via S-DPC. However, S-DPC is prohibitively complex for practical uses. Typically, an exhaustive search over all possible $\mathbf{Q}_1$ and 
$\mathbf{Q}_2$ satisfying  the constraints in \eqref{eq: mathmodel_const}  \cite{liu2010multiple} is used to get the capacity region. In \cite{fakoorian2013optimality}, 
a  GSVD-base precoder is proposed for this channel. Although its complexity is low, the rate region of GSVD-based precoding is far from the capacity region. 

In this letter, we show that the above secure MIMO-NOMA channel can be seen as two interwoven 
MIMO \textit{wiretap 
	channels}.
In one wiretap channel, user~$1$ is viewed as a legitimate user, while user~$2$ is an eavesdropper. The secrecy rate of this channel is obtained from \eqref{eq: mathmodel_confi1}.  In the second wiretap channel, the role of user~$1$ and user~$2$ is swapped, and the secrecy rate of this channel is obtained from \eqref{eq: mathmodel_confi2}.
Due to the symmetry of the channel, the rate region in \eqref{eq: mathmodel} can  be equivalently  obtained  by 
swapping the subscripts 1 and 2 in  \eqref{eq: mathmodel_confi1} and 
\eqref{eq: mathmodel_confi2} \cite[Corollary 1]{ekrem2012capacity}. Next, we design novel  precoding and power allocation schemes that achieve capacity region with reasonable complexity.

\section{Decomposing Secure MIMO-NOMA into Two MIMO Wiretap Channels}\label{sec:decomp}

In order to introduce new simpler solutions, in this section, 
we decompose the aforementioned secure MIMO-NOMA channel into two MIMO wiretap channels. This is done in three steps. First, 
similar to the BC channel, we split
the power between the two users. 
Then, we decouple the secure MIMO-NOMA channel into two MIMO wiretap channels to solve them separately, as described below.

\subsubsection*{Step 1}
Introducing power splitting 
factor $\alpha \in [0, 1]$, we dedicate a fraction $\alpha$ of 
the total power to  user~$1$   ($P_1=\alpha P$), and fraction $\bar \alpha$,  $\bar \alpha = 1- \alpha$,  to  user~$2$   ($P_2=\bar \alpha P$).

\subsubsection*{Step 2} 
We design secure precoding  for user~$1$ while treating  user~$2$ as an eavesdropper. Because 
\eqref{eq: 
mathmodel_confi1} is only controlled by the 
covariance matrix $\mathbf{Q}_1$, the problem can be 
seen as a wiretap channel under a transmit power $P_1$, which 
is 
\begin{subequations}\label{eq:R1}
	\begin{align}
	R_{1}({\alpha}) &= \max \limits_{\mathbf{Q}_1 \succeq 
	\mathbf{0}} 
\frac{1}{2}\log \frac{| \mathbf{I} + { \mathbf{H}_1 \mathbf{Q}_1 
		\mathbf{H}_1^T}|}{ |\mathbf{I} +  
	\mathbf{H}_2 \mathbf{Q}_1 \mathbf{H}_2^T|}, \label{eq:R1C1} \\
	&{\rm s.t.}\quad {\rm tr}(\mathbf{Q}_1)\leq P_1 = \alpha P. \label{eq:R1C2}
	\end{align}
\end{subequations}
This problem is 
now the 
well-known MIMO wiretap channel 
\cite{vaezi2017journal}, and standard MIMO wiretap solutions can be 
applied. 
 
\subsubsection*{Step 3}
We design secure precoding  for user~$2$ to maximize the rate of user~$2$ by allocating the remaining  
power,  
and using  $\mathbf{Q}^{*}_1$ obtained in \textit{Step 2} to \eqref{eq: 
mathmodel_confi2}. Thus,  \eqref{eq: 
mathmodel_confi2} is represented as
\begin{subequations} \label{eq:R2}
	\begin{align}
	R_{2}({\alpha}) &= \max \limits_{\mathbf{Q}_2 \succeq \mathbf{0}} 
	\bigg\{\frac{1}{2} \log \bigg|\mathbf{I} + \frac{\mathbf{H}_2 \mathbf{Q}_2 
	\mathbf{H}_2^T}{\mathbf{I} +\mathbf{H}_2 \mathbf{Q}^{*}_1 
	\mathbf{H}_2^T}\bigg|\notag \\
	& \quad \quad \quad \quad \quad \quad	-\frac{1}{2} \log \bigg|\mathbf{I} + \frac{\mathbf{H}_1 \mathbf{Q}_2 \mathbf{H}_1^T}{\mathbf{I} +\mathbf{H}_1 \mathbf{Q}^{*}_1 \mathbf{H}_1^T}\bigg|\bigg\}, \label{eq:R2C1}\\
	&{\rm s.t.}\quad {\rm tr}(\mathbf{Q}_2)\leq P_2 = (1-\alpha) P. \label{eq:R2C2}
	\end{align}
\end{subequations} 
Since  $\mathbf{Q}^{*}_1$  is given after solving \eqref{eq:R1}, in the following we show 
that the above problem can be seen as another wiretap 
channel where users~$2$ and $1$ 
are the legitimate user and eavesdropper, respectively. 
\begin{thm}\label{thm:theorem1}
	The above channel can be converted to a standard MIMO wiretap channel with 	\begin{subequations} \label{eq: sigmas1}
		\begin{align}
		\mathbf{H}^{\prime}_1 \triangleq 
		\mathbf{\Lambda}^{-\frac{1}{2}}_a\mathbf{V}^{T}_a\mathbf{H}_1, \\
		\mathbf{H}^{\prime}_2 \triangleq 
		\mathbf{\Lambda}^{-\frac{1}{2}}_b\mathbf{V}^{T}_b\mathbf{H}_2,
		\end{align}
	\end{subequations}
	in which $\mathbf{\Lambda}_a$ and 
	$\mathbf{V}_a$ are the eigenvalues and eigenvectors  of $\mathbf{I}+\mathbf{H}_1 
	\mathbf{Q}^{*}_1 \mathbf{H}_1^T$, and $\mathbf{\Lambda}_b$ and 
	$\mathbf{V}_b$ are the eigenvalues and eigenvectors of 
	$\mathbf{I}+\mathbf{H}_2 
	\mathbf{Q}^{*}_1 \mathbf{H}_2^T$.
\end{thm}
\begin{proof}
	Let us define  
	\begin{subequations} \label{eq: sigmas}
		\begin{align}
		\mathbf{\Sigma}_1 \triangleq \mathbf{I}+\mathbf{H}_1
	\mathbf{Q}^{*}_1
		\mathbf{H}_1^T=\mathbf{V}_a\mathbf{\Lambda}_a\mathbf{V}_a^T,\\
		\mathbf{\Sigma}_2 \triangleq \mathbf{I}+\mathbf{H}_2 
		\mathbf{Q}^{*}_1
		\mathbf{H}_2^T=\mathbf{V}_b\mathbf{\Lambda}_b\mathbf{V}_b^T,  
		\label{eq:sigma1} 
		\end{align}
	\end{subequations}
	Then, the rate for user~$2$ can be written as
		\begin{align} \label{eq: derivation}
		R_2(\alpha) &=\max \limits_{\mathbf{Q}_2 \succeq \mathbf{0}}  
		\frac{1}{2}\log \frac{|\mathbf{I} + \mathbf{H}_2 \mathbf{Q}_2 
		\mathbf{H}_2^T\mathbf{\Sigma}_2^{-1}|}{|\mathbf{I} +  \mathbf{H}_1 
		\mathbf{Q}_2 
				\mathbf{H}_1^T \mathbf{\Sigma}_1^{-1}|}    \notag \\
	& = \max \limits_{\mathbf{Q}_2 \succeq \mathbf{0}}  \frac{1}{2}\log 
		\frac{|\mathbf{I} +  \mathbf{H}_2 \mathbf{Q}_2 
		\mathbf{H}_2^T\mathbf{V}_b\mathbf{\Lambda}^{-1}_b\mathbf{V}_b^{T}|}{
			|\mathbf{I} +  \mathbf{H}_1 \mathbf{Q}_2 
			\mathbf{H}_1^T\mathbf{V}_a 
			\mathbf{\Lambda}^{-1}_a\mathbf{V}_a^{T}|}  \notag \\
		&\stackrel{(a)}{=} \max \limits_{\mathbf{Q}_2 \succeq \mathbf{0}}  \frac{1}{2}\log 
		\frac{|\mathbf{I} +  
		\mathbf{\Lambda}^{-\frac{1}{2}}_b\mathbf{V}^{T}_b\mathbf{H}_2 
		\mathbf{Q}_2 
		\mathbf{H}_2^T\mathbf{V}_b\mathbf{\Lambda}^{-\frac{1}{2}}_b|}{
			|\mathbf{I} +  
			\mathbf{\Lambda}^{-\frac{1}{2}}_a\mathbf{V}^{T}_a\mathbf{H}_1 
			\mathbf{Q}_2 
			\mathbf{H}_1^T\mathbf{V}_a\mathbf{\Lambda}^{-\frac{1}{2}}_a|} 
		\notag \\
		&=\max \limits_{\mathbf{Q}_2 \succeq \mathbf{0}}  \frac{1}{2}\log 
		\frac{|\mathbf{I} +  \mathbf{H}^{\prime}_2 \mathbf{Q}_2 
		\mathbf{H}^{\prime T}_2|}{
			|\mathbf{I} +  \mathbf{H}^{\prime}_1 \mathbf{Q}_2 
			\mathbf{H}^{\prime T}_1|},  
		\end{align}
in which $(a)$  holds because $\det(\mathbf{I} + \mathbf{A}\mathbf{B}) = 
	\det(\mathbf{I} + \mathbf{B}\mathbf{A})$ and 
	$\mathbf{\Lambda}_a$  and $\mathbf{\Lambda}_b$ are diagonal matrices. 
\end{proof}

In view of  \eqref{eq: derivation}, it is seen that like \eqref{eq:R1C1},  \eqref{eq:R2C1} is the rate for a MIMO wiretap channel with  channels $\mathbf{H}^{\prime}_2 $ for the legitimate user and $\mathbf{H}^{\prime}_1 $ for the eavesdropper.  
\section{Secure Precoding  and Power Allocation}

In this section, we propose new linear precoding and power allocation strategies to secure the MIMO-NOMA channel. 
In light of our decomposition in the previous section, we have two 
MIMO wiretap channels and 
thus standard MIMO wiretap solutions can be applied to design covariance matrices $\mathbf{Q}_1$  and $\mathbf{Q}_2$.  
One fast approach  is rotation 
based linear precoding \cite{vaezi2017journal}. In this method, the covariance matrix 
$\mathbf{Q}_1$ 
is eigendecomposed into one rotation matrix $\mathbf{V}_1$ and one power allocation matrix $\mathbf{\Lambda}_1$ \cite{vaezi2017journal, vaezi2019rotation} as
\begin{align}\label{eq_eig1}
\mathbf{Q}_1= \mathbf{V}_1 \mathbf{\Lambda}_1 \mathbf{V}_1^T.
\end{align}
Consequently, the secrecy capacity of 
user 1 is
\begin{subequations} \label{eq:r1star}
	\begin{align} 	
	R_{1}(\alpha)
	&=\max \limits_{\mathbf{Q}_1 \succeq \mathbf{0}} \frac{1}{2}\log \frac{|\mathbf{I} 
		+ { \mathbf{H}_1 \mathbf{V}_1\mathbf{\Lambda}_1\mathbf{V}_1^{T} 
			\mathbf{H}_1^T}|}{|  \mathbf{I} 
		+ \mathbf{H}_2 
		\mathbf{V}_1\mathbf{\Lambda}_1\mathbf{V}_1^{T} \mathbf{H}_2^T|}, 
	\label{eq:r1star1}\\
	&{\rm s.t.}\quad \sum_{k = 1}^{n_t} \lambda_{1k} \leq P_1 = 
	\alpha P, \label{eq:r1starconst}
	\end{align}
\end{subequations}
in which $\lambda_{1k}$, $k= \{1,\dots,n_t\}$, is a diagonal 
element of matrix $\mathbf{\Lambda}_1 = {\rm diag}(\lambda_{11}, \dots, 
\lambda_{1n_t})$.
The rotation matrix 
$\mathbf{V}_1$ can be obtained by
\begin{align}\label{eq_Vnbyn_}
\mathbf{V}_1=\prod_{i=1}^{n_t-1}\prod_{j=i+1}^{n_t} \mathbf{V}_{ij},
\end{align}
in which the basic rotation matrix $ \mathbf{V}_{ij} $ is a Givens matrix 
\cite{matrixbook} which is an identity matrix except that its elements in the 
$i$th row and $j$th column, i.e., $v_{ii}$, $v_{ij}$, 
$v_{ji}$, and $v_{jj}$ are replaced by 
\begin{align}\label{eq_VnDsub}
\left[
\begin{matrix}
v_{ii}	&v_{ij}\\
v_{ji}	&v_{jj}
\end{matrix}
\right]
=\left[
\begin{matrix}
\cos\theta_{1ij}	&-\sin \theta_{1ij}\\
\sin\theta_{1ij}	&\cos \theta_{1ij}
\end{matrix}
\right],
\end{align}
where  
$\theta_{1ij}$ is rotation angle corresponding to the rotation 
matrix $ \mathbf{V}_{ij} $.
Then, we will optimize the new parameterized nonconvex problem 
numerically to obtain the solution $\mathbf{Q}^{*}_1$ with respect to 
rotation angles and power allocation parameters.  

Similarly, covariance matrix $\mathbf{Q}_2$ can be written by rotation method as $\mathbf{Q}_2= \mathbf{V}_2 \mathbf{\Lambda}_2 \mathbf{V}_2^T$,
where the rotation matrix $\mathbf{V}_2$ is defined similar to $\mathbf{V}_1$ in \eqref{eq_Vnbyn_} with its rotation angles are
$\theta_{2ij}$. Therefore, the optimization problem for $R_{2}(\alpha)$ becomes
\begin{subequations} \label{eq:r2star}
	\begin{align} 	
	R_{2}(\alpha) 
	&=\max \limits_{\mathbf{Q}_2 \succeq \mathbf{0}} \frac{1}{2}\log 
	\frac{|\mathbf{I} + { \mathbf{H}^{\prime}_1 
			\mathbf{V}_2\mathbf{\Lambda}_2\mathbf{V}_2^{T} \mathbf{H}^{\prime 
				T}_1}|}{|\mathbf{I} +  \mathbf{H}^{\prime }_2 
		\mathbf{V}_2\mathbf{\Lambda}_2\mathbf{V}_2^{T} \mathbf{H}^{\prime T}_2|} 
	\label{eq:r2star2},\\
	&{\rm s.t.}\quad \sum_{k = 1}^{n_t} \lambda_{2k} \leq P_2 = (1-\alpha) 
	P,\label{eq:r2starconst}
	\end{align}
\end{subequations}
in which $\lambda_{2k}$ is the $k$th  diagonal element of $\mathbf{\Lambda}_2$. This problem is again similar to \eqref{eq:r1star}.


To solve the new parameterized problems in \eqref{eq:r1star} and \eqref{eq:r2star}  to find new parameters $\lambda_{1k}$, $\theta_{1ij}$ and  $\lambda_{2k}$, $\theta_{2ij}$ (instead of directly finding $\mathbf{Q}_1$ and $\mathbf{Q}_2$ in  \eqref{eq: mathmodel}), various numerical approaches 
such as  Matlab \texttt{fmincon} 
can be used.  
In this paper, for a fixed $\alpha \in [0, 1]$,  we use Broyden-Fletcher-Goldfarb-Shanno (BFGS) method together with the interior-point method 
(IPM) 
\cite{nocedal2006numerical}.  IPM transfers the constraints into 
an unconstrained problem, and BFGS is a quasi-Newton iterative 
method for nonlinear optimization. 
The algorithm is elaborated  in Algorithm~\ref{alg:randomalgorithm}.

In the power splitting method, we introduced in Section~\ref{sec:decomp}, and new optimization problems in  \eqref{eq:r1star} and \eqref{eq:r2star}, for each $\alpha$,  we solve for $\mathbf{Q}_1^{*}$ and $\mathbf{Q}_2^{*}$ (and  thus $R_1^{*}(\alpha)$ and $R_2^{*}(\alpha)$) step by step. This simplifies the problem  but may result in sub-optimal rate region. Moreover, the order of optimization (first $R_1^{*}$ then  $R_2^{*}$) will affect the solution.

\begin{figure}[t]
	\centering
	\includegraphics[width=0.34\textwidth]{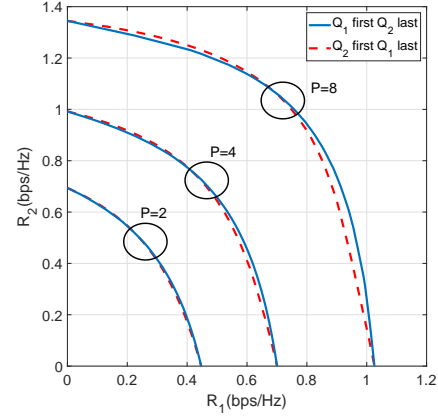}
	\caption{The effect of the order of precoding on the secure rate region for different values of transmit power ($P$). 
	}
	\label{fig:r0122}
\end{figure}

\begin{algorithm}[!t]
	\caption{Rotation Based Precoding and Power Allocation}\label{alg:randomalgorithm}
	\begin{algorithmic}[1]
		\STATE	inputs: $P$, $n_t$, $n_1$, $n_2$, $\sigma$;
		\FOR {$\alpha$ from $0$ to $1$ with a step $\sigma$}
		\CASE {1}
		\STATE	$P_1 = \alpha P $;
		\STATE solve for $\lambda_{1k}$ and $\theta_{1ij}$, in \eqref{eq:r1star} 
		using BFGS and IPM;
		\STATE generate   {\color{black}$\mathbf{Q}^{*}_1$} by \eqref{eq_eig1} and 
		\eqref{eq_Vnbyn_}--\eqref{eq_VnDsub};
		\STATE calculate {\color{black}$R^{*}_{1}(\alpha)$} by inserting $\mathbf{Q}_1^{*}$ in 
		\eqref{eq:r1star1};
		\STATE	$P_2 = (1-\alpha) P $;
		\STATE calculate  $\mathbf{H}^{\prime}_1$ and $ \mathbf{H}^{\prime}_2$ via 
		\eqref{eq: sigmas1};
		\STATE solve for $\lambda_{2k}$  and  $\theta_{2ij}$ in 
		\eqref{eq:r2star} 
		using BFGS  and IPM;
		\STATE calculate {\color{black}$R^{*}_{2}(\alpha)$} by inserting $\mathbf{Q}_2^{*}$ in 
		\eqref{eq:r2star2};
		\STATE obtain $(R^{*}_{1}(\alpha), R^{*}_{2}(\alpha))$; 
		\ENDCASE
		\CASE {2}
		\STATE swap all subscripts of 1 and 2 in  \eqref{eq: mathmodel} and 
		case 1;
		\STATE repeat case 1 and obtain $(\bar{R}^{*}_{1}(\alpha), 
		\bar{R}^{*}_{2}(\alpha))$;
		\ENDCASE
		\ENDFOR
		\STATE Obtain the secrecy region {\color{black}$\mathcal{R}_s$} via 
		Lemma~\ref{lem1}.
	\end{algorithmic}
\end{algorithm}

Alternatively, we can first solve for $\mathbf{Q}_2^{*}$ followed by $\mathbf{Q}_1^{*}$ (i.e., first $R_2^{*}$ then $R_1^{*}$). We represent this solution $(\bar{R}^{*}_{1}(\alpha), 
\bar{R}^{*}_{2}(\alpha))$ in Algorithm~\ref{alg:randomalgorithm}. In general, changing the order of optimization will result in a different rate region. To show how the  order of precoding can change the achievable rate region, we demonstrate an example in Fig.~\ref{fig:r0122} with different powers $P=2, 4, 8$, where the 
	channels are
	\begin{align}
	\mathbf{H}_1&=\left[
	\begin{matrix}
	0.125 &	0.821&	0.087 \\
	0.383 &	0.261&	0.037
	\end{matrix}\right], \notag\\
	\mathbf{H}_2&=\left[
	\begin{matrix}
	0.384&	0.703&	0.849
	\end{matrix}\right]. \notag
	\end{align} Thus, the \textit{convex hull} of the two solutions with different orders may enlarge the achievable rate region.

\begin{figure*}[t]
	\centering
	\begin{minipage}[t]{0.32\textwidth}
		\centering
		\includegraphics[width=6cm]{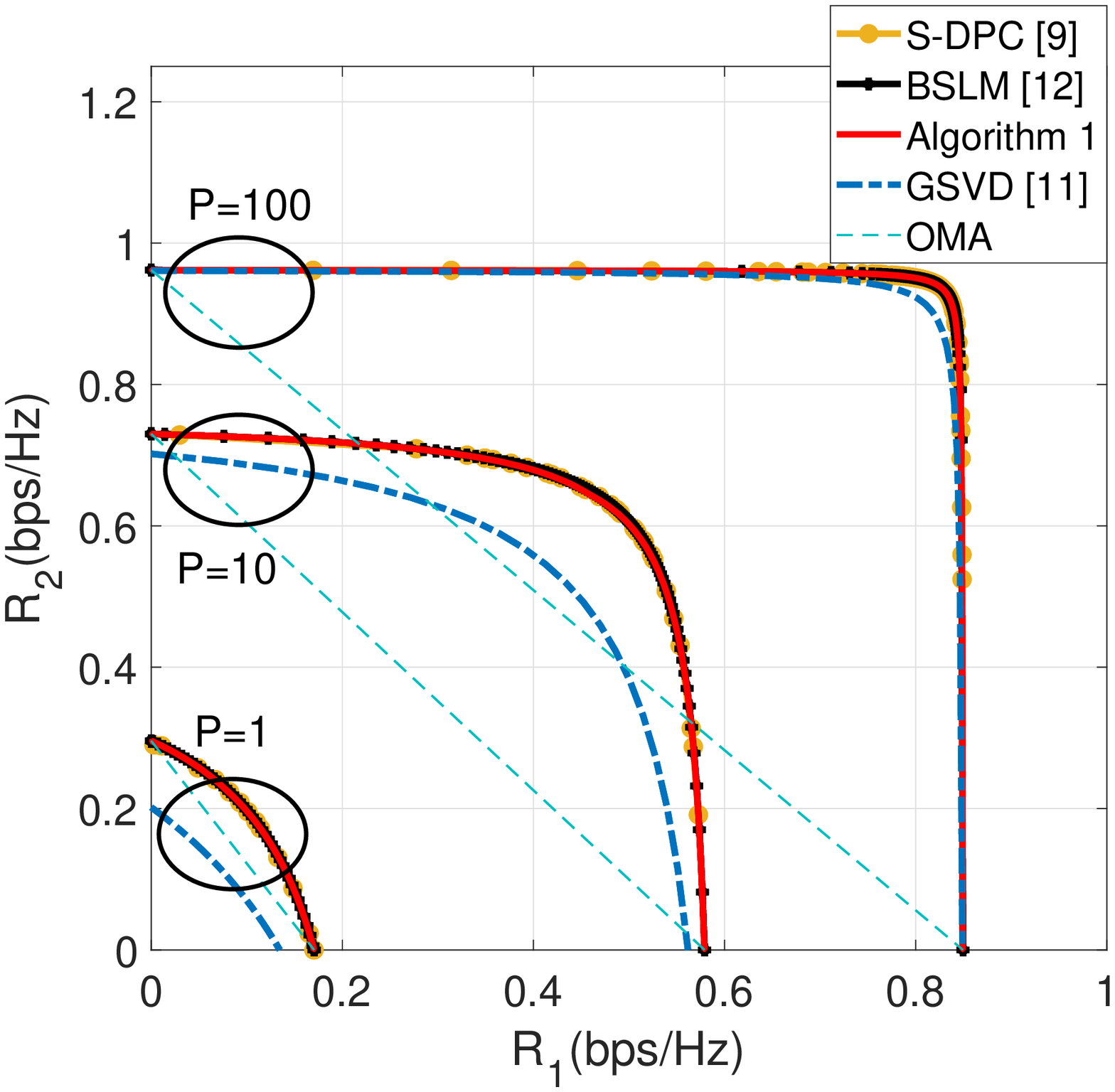}
		\caption{Secrecy rate regions for $n_t=n_1=n_2=2$ with different powers.}
		\label{fig:r1r2}
	\end{minipage} 
	\begin{minipage}[t]{0.32\textwidth}
		\centering
		\includegraphics[width=6cm]{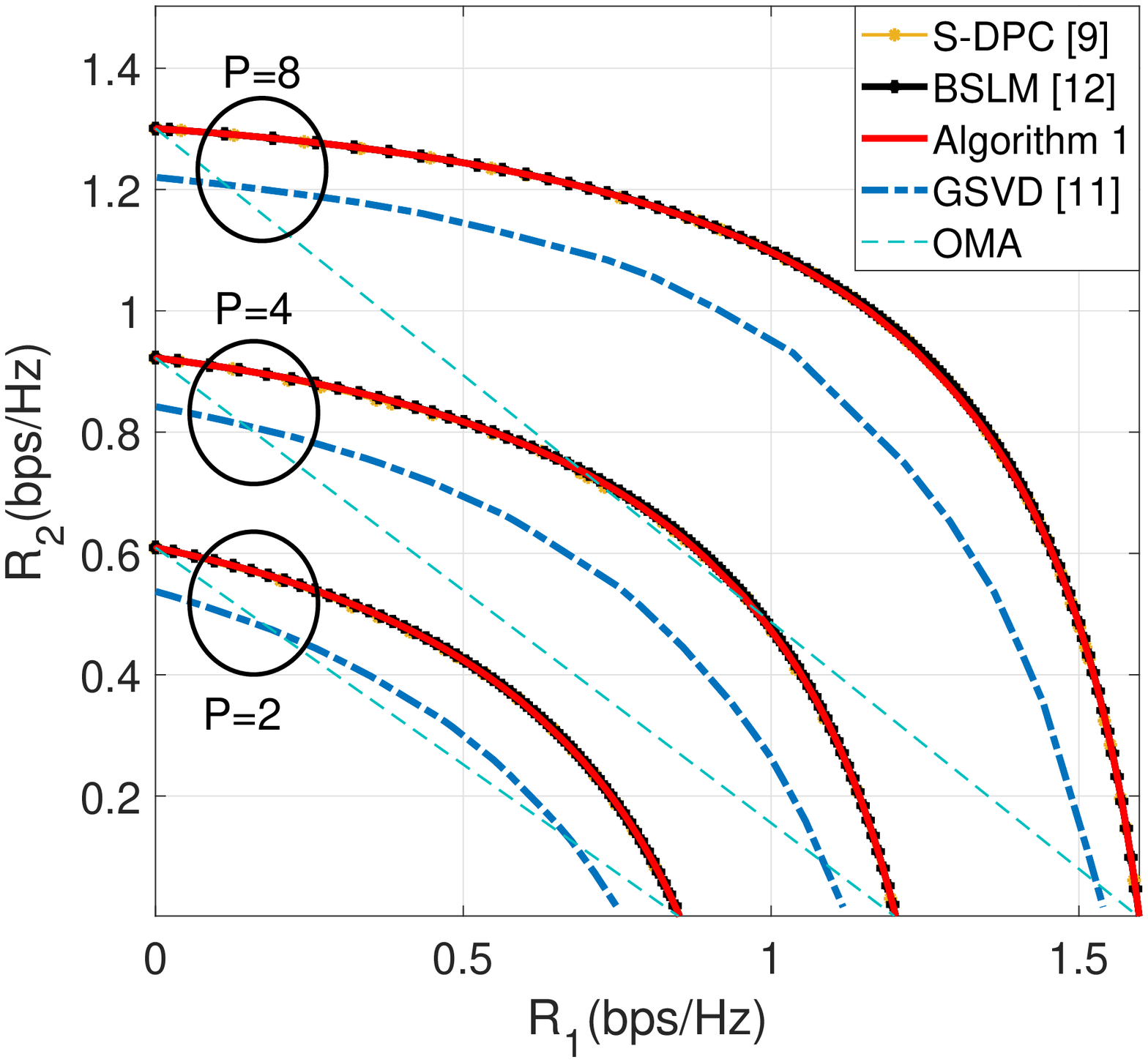}
		\caption{Secrecy rate regions for $n_t=3, n_1=n_2=2$ with different powers. }
		\label{fig:r1r2_1}
	\end{minipage} 
	\begin{minipage}[t]{0.32\textwidth}
		\centering
		\includegraphics[width=6cm]{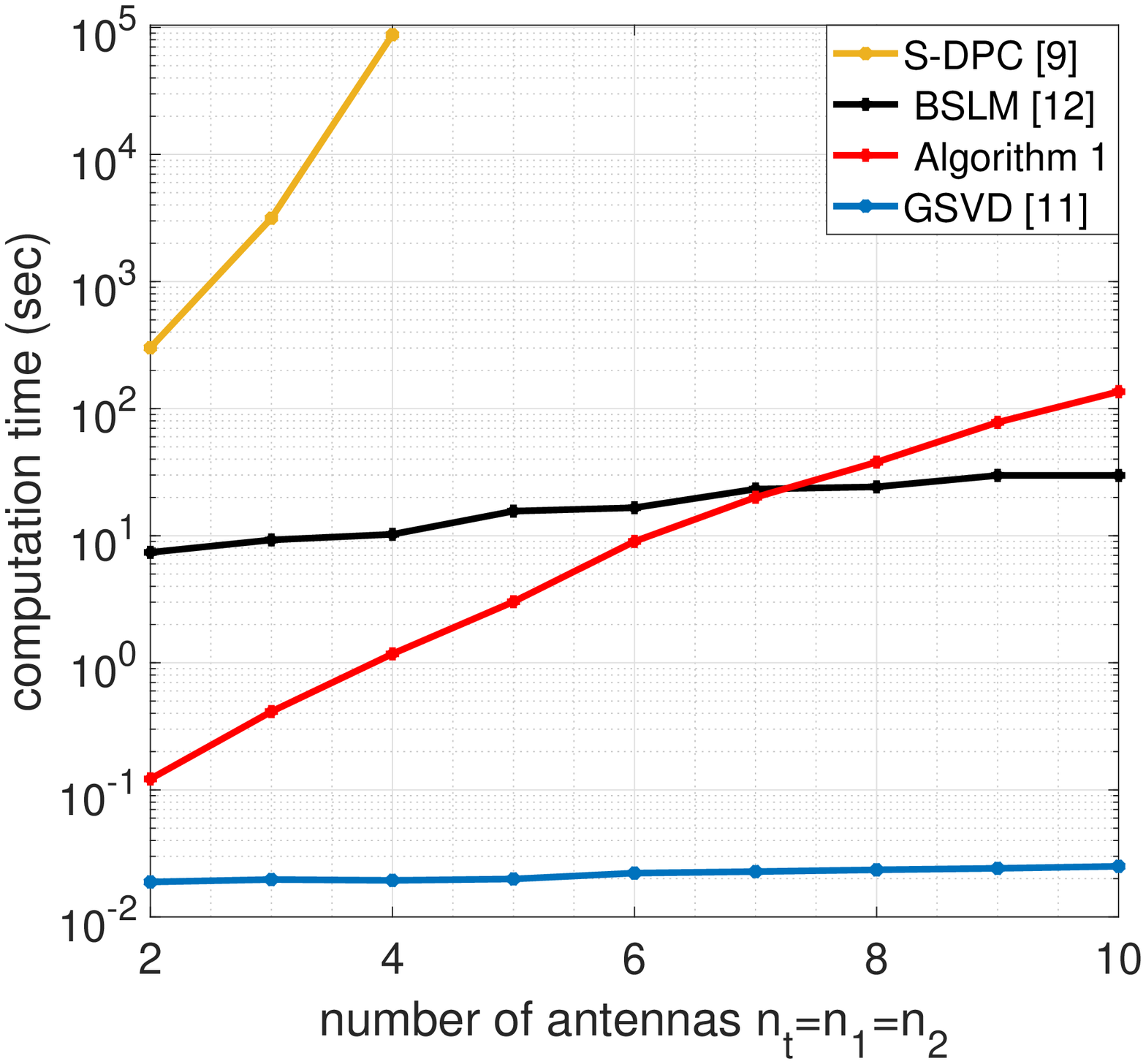}
		\caption{Computational time  versus number of antennas for $P=30$.}
		\label{fig:r1r2_2}
	\end{minipage}
\end{figure*}

\begin{lem} \label{lem1}
	The achievable secrecy region $\mathcal{R}_s$ for the secure MIMO-NOMA channel 
	under 
	total  power constraint is the convex hull  of 
	all rate points 
	\begin{align}
	\mathcal{R}_s=
	\bigcup_{\substack{0 \leq \alpha \leq 1}} \bigg\{ (R^{*}_{1}(\alpha), 
	R^{*}_{2}(\alpha)) \cup (\bar{R}^{*}_{1}(\alpha), \bar{R}^{*}_{2}(\alpha))\bigg\},
	\end{align}
	where $(R^{*}_{1}(\alpha), R^{*}_{2}(\alpha))$ is 
	obtained by  precoding for user~$1$  then user~$2$, whereas
	$(\bar{R}^{*}_{1}(\alpha), \bar{R}^{*}_{2}(\alpha))$ is obtained by  precoding in the reverse order (first user~$2$,  then user~$1$).	
\end{lem}

	\textit{Remark 1 (Complexity Analysis):} 	
	The BFGS algorithm yields the complexity 
	of $\mathcal{O}(n^2)$ \cite{nocedal2006numerical}, where $n$ is the size of input 
	variables which in our case is the number of rotation parameters, i.e., $n=\frac{(n_t+1)n_t}{2}$. 
	On the other hand, the computation of matrix  multiplications and matrix inverse yield the complexity of $\mathcal{O}(L^3)$ where $L= \max(n_t, n_1, n_2)$. Thus, the overall complexity of Algorithm~\ref{alg:randomalgorithm} is 
	$\mathcal{O}(\frac{n_t^4+L^3}{\sigma})$   and  
	$\sigma$ is the search step of the power fraction $\alpha$.
	The weighted sum-rate has the complexity of $\mathcal{O}(\frac{L^3}{\epsilon \sigma}\log({1}/{ \epsilon}))$ with a search over the weight \cite{park2015weighted}, and $\epsilon$ is the convergence tolerance of algorithm.
	The GSVD-based precoding \cite{fakoorian2013optimality} 
	has $\mathcal{O}(\frac{L^3}{\sigma}+\frac{L}{\sigma}\log(1/\epsilon))$, in which the search step $\sigma$ over power  comes from \cite[Corollary ~1]{fakoorian2013optimality}.
	The computational complexity of the exhaustive search for S-DPC is exponential  in $L$ \cite{liu2010multiple}.

\section{Numerical results}\label{sec:Section IV}

In this section, we compare the rate region   
$(R_1, R_2)$ of the proposed method with the secrecy capacity region obtained 
by S-DPC \cite{liu2010multiple}, weighted secrecy sum-rate maximization using the 
block successive lower-bound maximization (BSLM)  
\cite{park2015weighted}, and
GSVD-based precoding 
\cite{fakoorian2013optimality}. The capacity region is obtained by an exhaustive search 
over all  input covariance matrices whereas
 BSLM maximization is achieved by an iterative algorithm, which successively 
optimizes a sequence of approximated functions using binary search 
\cite{park2015weighted} to reach higher rates. 
 GSVD-based precoding is an 
orthogonal channel assignment, and orthogonal multiple access (OMA) is achieved by the time-sharing between the two extreme points (two wiretap channels) which realizes the same task in two orthogonal time slots controlled by the time-sharing fraction.

The 
channel matrices $\mathbf{H}_1$ and 
$\mathbf{H}_2$ are generated randomly 
\begin{align}
\mathbf{H}_1=\left[
\begin{matrix}
0.783	&0.590 \\
0.734 &0.092
\end{matrix}\right], \quad
\mathbf{H}_2=\left[
\begin{matrix}
0.244&0.617\\
0.947&0.807
\end{matrix}\right], \notag
\end{align}
for $P=1, 10, 100$. The search step $\epsilon$ for 
Algorithm~\ref{alg:randomalgorithm} 
and GSVD-based 
precoding \cite{fakoorian2013optimality} is set as $0.05$ identically.  As 
Fig.~\ref{fig:r1r2} illustrates,  the proposed algorithm achieves larger rate regions 
compared to  GSVD-based  precoding 
\cite{fakoorian2013optimality}, and almost identical to the capacity  and 
BSLM maximization \cite{park2015weighted}. 
{ \color{black}In the MISO case of $\mathbf{H}_1=[1.5 \; 0]$ and  $\mathbf{H}_1=[1.801 \; 
0.871]$, and $P=10, 100$, and $1000$,  we found that our proposed algorithm can reach 
the same secrecy region and obtained  exactly the same figure in \cite[Fig.  
2]{park2015weighted}.}

Fig.~\ref{fig:r1r2_1} shows  secrecy 
rate regions for $n_t=3$ with $P=2, 4, 8$, and 
$\sigma=0.05$, where the 
channels are
\begin{align}
\mathbf{H}_1&=\left[
\begin{matrix}
0.813 &	0.232&	0.085 \\
0.842 &	0.130&	0.203
\end{matrix}\right], \notag\\
\mathbf{H}_2&=\left[
\begin{matrix}
0.315&	0.769&	0.294 \\
0.025&  0.271& 0.281
\end{matrix}\right]. \notag
\end{align}

Again, the proposed method largely outperforms  GSVD-based 
precoding, and achieve the same rate region compared with the BSLM method. 

Although both BSLM and the proposed method 
are very close to the secrecy capacity, the computation costs of the two methods are remarkably different, particularly for practical numbers of antennas.  
Table.~\ref{tab_rate_nt2}, Table.~\ref{tab_rate_nt3}, and Fig.~\ref{fig:r1r2_2} show the 
execution time for all three precoding methods over 100 random 
channel realizations.  BSLM takes much  higher time to reach the same rate region as our proposed method in 
Algorithm~\ref{alg:randomalgorithm}, as BSLM is achieved by an 
iterative algorithm. 
Our proposed method  
can save time ten to a hundred times especially for portable devices equipped with a few antennas and can strike a balance between performance and computation complexity. 

	  For massive MIMO-NOMA, the number of optimization parameters in Algorithm~\ref{alg:randomalgorithm} become considerable as $n_t$ becomes very large. The solution proposed in this letter is not meant for such scenarios.  One can, however, reduce the complexity for large values of $n_t$ by applying other wiretap solutions such as  alternating optimization \cite{li2013transmit}. On the other hand, other massive MIMO solutions, such as  \cite{asaad2018optimal,wu2016secure,zeng2019securing} may be more effective.

\section{Conclusions}\label{sec:Section V}

A novel linear precoding has been proposed for secure transmission over 
MIMO-NOMA networks to prevent users from eavesdropping each other. The proposed approach decomposes the two-user MIMO-NOMA channel into two MIMO 
wiretap channels via splitting the base station power between the two users and modifying the channel corresponding to one of the users to make it a wiretap channel in effect. This approach  achieves a significantly higher secure rate region compared to existing linear precoding, and has an acceptable computational complexity.

\begin{table}[t]
	\caption{{Execution time (ms) for $n_t=2$ and $P=30$.}}
	\label{tab_rate_nt2}
	\centering
	\begin{tabular}{|c|c|ccccc|}
		\hline
		\multicolumn{2}{|c|}{\multirow{2}{*}{BSLM 
				\cite{park2015weighted}}} & 
		\multicolumn{5}{c|}{$n_2$}  \\ \cline{3-7}
		\multicolumn{2}{|c|}{}  & 1    & 2    & 3    & 4    &5\\ \hline
		\multirow{5}{*}{$n_1$}
	&1&2063.0&    3428.2&      3774.0&      4904&    4516.6\\
	&2&3735.1&    6524.8&    7411.6&    8630.8&    9400.4\\
	&3&4814.3&    7158.5&    9675.9&   10310.4&   10411.4\\
	&4&5426.1&      9380.0&   11285.4&   13446.8&   13086.3\\
	&5&6776.4&     10931.0&   13572.1&   15300.5&   16959.1\\  \hline
		\multicolumn{2}{|c|}{\multirow{2}{*}{Algorithm~\ref{alg:randomalgorithm}}}
		&		\multicolumn{5}{c|}{$n_2$}                          \\ \cline{3-7}
		\multicolumn{2}{|c|}{} & 1      & 2      & 3      & 4      & 5 \\   \hline
		\multirow{5}{*}{$n_1$} 
		&1&113.7&112.2&101.8&102.3&  95.3\\
		&2&108.2&105.5&99.9&  98.1&  93.5\\
		&3&95.0& 100.2& 101.9&  99.0&    96.0\\
		&4&94.6&  98.9& 106.6& 99.1& 101.2\\
		&5&94.2&  92.9& 98.4&  96.8&  95.6\\ \hline
		\multicolumn{2}{|c|}{\multirow{2}{*}{GSVD
				\cite{fakoorian2013optimality}}} & 
		\multicolumn{5}{c|}{$n_2$}                                 \\ \cline{3-7}
		\multicolumn{2}{|c|}{}                         & 1        & 2        & 3      & 4       
		& 5             \\ \hline
		\multirow{5}{*}{$n_1$}  
		&1&16.9&  17.0&  16.5&  17.0&  16.6\\
		&2&16.8&  17.0&  17.3&  16.9&  17.1\\
		&3&16.6&  17.3&  17.2&  17.6&  17.8\\
		&4&16.8&  17.2&  18.6&  17.8&  19.1\\
		&5&18.1&  17.3&  17.6&  18.0&  17.7\\
		\hline
	\end{tabular}
\end{table}

\begin{table}[t]
	\caption{Execution time (ms) for $n_t=3$ and $P=30$.}
	\label{tab_rate_nt3}
	\centering
	\begin{tabular}{|c|c|ccccc|}
		\hline
		\multicolumn{2}{|c|}{\multirow{2}{*}{BSLM 
				\cite{park2015weighted}}} & 
		\multicolumn{5}{c|}{$n_2$}  \\ \cline{3-7}
		\multicolumn{2}{|c|}{}  & 1    & 2    & 3    & 4    &5\\ \hline
		\multirow{5}{*}{$n_1$}
		&1&2364.8& 3029.8& 4218.9& 4893.1& 4776.5\\
		&2&3386.5& 4612.6& 6138.8&   7504.0& 9829.3\\
		&3&5304.3& 8190.1&10169.9&11428.9&  12051.0\\
		&4&5432.4& 7982.4&  10363.0&13440.3&13744.5\\
		&5&6708.0& 8450.5&11337.1&13851.8&15552.5\\ \hline
		\multicolumn{2}{|c|}{\multirow{2}{*}{Algorithm~\ref{alg:randomalgorithm}}}
		&		\multicolumn{5}{c|}{$n_2$}                          \\ \cline{3-7}
		\multicolumn{2}{|c|}{}                         & 1      & 2      & 3      & 4      & 
		5        \\ \hline
		\multirow{5}{*}{$n_1$} 
		&1&215.5&  503.2&    476.0&  448.6&  409.1\\
		&2&493.2&  399.9&  375.7&  379.9&  371.8\\
		&3&478.8&  381.7&    354.0&  392.6&    396.0\\
		&4&428.5&    364.0&  379.8&  407.9&  388.1\\
		&5&460.7&  391.3&  377.8&  370.6&  371.9\\ \hline
		\multicolumn{2}{|c|}{\multirow{2}{*}{GSVD
				\cite{fakoorian2013optimality}}} & 
		\multicolumn{5}{c|}{$n_2$}                                 \\ \cline{3-7}
		\multicolumn{2}{|c|}{}                         & 1        & 2        & 3      & 4       
		& 5             \\ \hline
		\multirow{5}{*}{$n_1$}  
		&1&17.1&   19.6&   17.8&   18.1&   18.6\\
		&2&19.5&   18.9&   18.2&   19.4&   19.6\\
		&3&18.2&     18.0&   18.4&     22.0&  21.6\\
		&4&18.7&   19.4&   21.1&   23.4&   22.6\\
		&5&23.2&   22.2&   21.2&   22.5&   23.3\\
		\hline
	\end{tabular}
\end{table}

\typeout{}

\balance
\bibliography{servicejournal}

\begin{thebibliography}{10}

\bibitem{vaezi2019interplay}
M.~Vaezi, G.~A.~A. Baduge, Y.~Liu, A.~Arafa, F.~Fang, and Z.~Ding, ``Interplay
  between {NOMA} and other emerging technologies: A survey,'' {\em IEEE
  Transactions on Cognitive Communications and Networking}, vol.~5, no.~4,
  pp.~900--919, 2019.

\bibitem{liu2017enhancing}
Y.~Liu, Z.~Qin, M.~Elkashlan, Y.~Gao, and L.~Hanzo, ``Enhancing the physical
  layer security of non-orthogonal multiple access in large-scale networks,''
  {\em IEEE Transactions on Wireless Communications}, vol.~16, no.~3,
  pp.~1656--1672, 2017.

\bibitem{chen2018physical}
J.~Chen, L.~Yang, and M.-S. Alouini, ``Physical layer security for cooperative
  {NOMA} systems,'' {\em IEEE Transactions on Vehicular Technology}, vol.~67,
  no.~5, pp.~4645--4649, 2018.

\bibitem{zheng2018secure}
B.~Zheng, M.~Wen, C.-X. Wang, X.~Wang, F.~Chen, J.~Tang, and F.~Ji, ``Secure
  {NOMA} based two-way relay networks using artificial noise and full duplex,''
  {\em IEEE Journal on Selected Areas in Communications}, vol.~36, no.~7,
  pp.~1426--1440, 2018.

\bibitem{arafa2019secure}
A.~Arafa, W.~Shin, M.~Vaezi, and H.~V. Poor, ``Secure relaying in
  non-orthogonal multiple access: Trusted and untrusted scenarios,'' {\em IEEE
  Transactions on Information Forensics and Security}, vol.~15, no.~1,
  pp.~210--222, 2019.

\bibitem{lv2018secure}
L.~Lv, Z.~Ding, Q.~Ni, and J.~Chen, ``Secure {MISO-NOMA} transmission with
  artificial noise,'' {\em IEEE Transactions on Vehicular Technology}, vol.~67,
  no.~7, pp.~6700--6705, 2018.

\bibitem{feng2019beamforming}
Y.~Feng, S.~Yan, Z.~Yang, N.~Yang, and J.~Yuan, ``Beamforming design and power
  allocation for secure transmission with {NOMA},'' {\em IEEE Transactions on
  Wireless Communications}, vol.~18, no.~5, pp.~2639--2651, 2019.

\bibitem{zeng2019securing}
M.~Zeng, N.-P. Nguyen, O.~A. Dobre, and H.~V. Poor, ``Securing downlink massive
  {MIMO-NOMA} networks with artificial noise,'' {\em IEEE Journal of Selected
  Topics in Signal Processing}, vol.~13, no.~3, pp.~685--699, 2019.

\bibitem{liu2010multiple}
R.~Liu, T.~Liu, H.~V. Poor, and S.~Shamai, ``Multiple-input multiple-output
  {Gaussian} broadcast channels with confidential messages,'' {\em IEEE
  Transactions on Information Theory}, vol.~56, no.~9, pp.~4215--4227, 2010.

\bibitem{vaezi2019noma}
M.~Vaezi and H.~V. Poor, ``{NOMA: An} information-theoretic perspective,'' in
  {\em Multiple Access Techniques for 5G Wireless Networks and Beyond},
  pp.~167--193, Springer, 2019.

\bibitem{fakoorian2013optimality}
S.~A.~A. Fakoorian and A.~L. Swindlehurst, ``On the optimality of linear
  precoding for secrecy in the {MIMO} broadcast channel,'' {\em IEEE Journal on
  Selected Areas in Communications}, vol.~31, no.~9, pp.~1701--1713, 2013.

\bibitem{park2015weighted}
D.~Park, ``Weighted sum rate maximization of {MIMO} broadcast and interference
  channels with confidential messages,'' {\em IEEE Transactions on Wireless
  Communications}, vol.~15, no.~3, pp.~1742--1753, 2015.

\bibitem{ekrem2012capacity}
E.~Ekrem and S.~Ulukus, ``Capacity region of {Gaussian} {MIMO} broadcast
  channels with common and confidential messages,'' {\em IEEE Transactions on
  Information Theory}, vol.~58, no.~9, pp.~5669--5680, 2012.

\bibitem{vaezi2017journal}
M.~Vaezi, W.~Shin, and H.~V. Poor, ``Optimal beamforming for {Gaussian MIMO}
  wiretap channels with two transmit antennas,'' {\em IEEE Transactions on
  Wireless Communications}, vol.~16, no.~10, pp.~6726--6735, 2017.

\bibitem{vaezi2019rotation}
M.~Vaezi, Y.~Qi, and X.~Zhang, ``A rotation-based precoding for {MIMO}
  broadcast channels with integrated services,'' {\em IEEE Signal Processing
  Letters}, vol.~26, no.~11, pp.~1708--1712, 2019.

\bibitem{matrixbook}
G.~H. Golub and C.~F. Van~Loan, {\em Matrix computations}, vol.~3.
\newblock JHU press, 2012.

\bibitem{nocedal2006numerical}
J.~Nocedal and S.~Wright, {\em Numerical optimization}.
\newblock Springer Science \& Business Media, 2006.

\bibitem{li2013transmit}
Q.~Li, M.~Hong, H.-T. Wai, Y.-F. Liu, W.~Ma, and Z.-Q. Luo, ``Transmit
  solutions for {MIMO} wiretap channels using alternating optimization,'' {\em
  IEEE Journal on Selected Areas in Communications}, vol.~31, no.~9,
  pp.~1714--1727, 2013.

\bibitem{asaad2018optimal}
S.~Asaad, A.~Bereyhi, A.~M. Rabiei, R.~R. M{\"u}ller, and R.~F. Schaefer,
  ``Optimal transmit antenna selection for massive {MIMO} wiretap channels,''
  {\em IEEE Journal on Selected Areas in Communications}, vol.~36, no.~4,
  pp.~817--828, 2018.

\bibitem{wu2016secure}
Y.~Wu, R.~Schober, D.~W.~K. Ng, C.~Xiao, and G.~Caire, ``Secure massive {MIMO}
  transmission with an active eavesdropper,'' {\em IEEE Transactions on
  Information Theory}, vol.~62, no.~7, pp.~3880--3900, 2016.

\end{thebibliography}
\bibliographystyle{ieeetr}

\end{document}